\documentclass{article}
\usepackage{graphicx}

\usepackage{float}% Required for inserting images
\usepackage{tikz}
\usepackage{subcaption}
\usepackage{graphicx}
\usepackage{amsmath,amssymb,amsthm}
\usepackage{mathrsfs}
\usepackage{pgfplots}
\theoremstyle{plain}
\newtheorem{theorem}{Theorem}[section]

\theoremstyle{definition}

\theoremstyle{remark}

\usepackage[backend=biber,style=alphabetic]{biblatex}
\newcommand{\Z}{\mathbb{Z}}

\newcommand{\F}{\mathbb{F}}

\usepackage{hyperref}
\title{The Value of Duplicate Measurements When Decoding Second-Order Reed-Muller Codes}
\author{Ziang Gao, Robert Calderbank}
\date{August 2026}
\begin{document}

\maketitle
\textbf{Abstract:} We compare the performance of the RPA and CHIRRUP decoding algorithms for the
second-order Reed-Muller code $RM(2,m)$. The RPA algorithm is designed to recover a noisy quadratic form $Q_M$ associated with a binary skew-symmetric $(m+1)\times(m+1)$ matrix $M$, whereas the CHIRRUP algorithm is designed to recover a noisy $\mathbb Z_4$-valued quadratic form $T_P$ associated with a binary symmetric matrix $P$. We describe how the Gray map connects the evaluation vectors for $Q_M$ and $T_P$, inducing a correspondence between $M$ and $P$ that is not linear, but is rank-preserving in the sense that if $M$ has rank $m+1-2i$, then $P$ has rank $m+1-2i$ or $m-2i$.\\
\hfill\break
We analyze how the Euclidean distance between evaluation vectors is determined by the rank of the difference between the corresponding matrices $M$ or $P$. We introduce Delsarte-Goethals sets $DG(m,r)$,
$r=0,1,\ldots,\frac{m-1}{2}$, which are sets of binary symmetric $m\times m$ matrices with the property that the difference of two distinct matrices in $DG(m,r)$ has rank at least $m-2r$. We use these sets to study how the performance of RPA and CHIRRUP depends on the availability of duplicate measurements, as determined by the rank and algebraic structure of the matrices in an ensemble.\\
\hfill\break
We demonstrate through simulations that when restricted to Delsarte-Goethals code ensembles, CHIRRUP makes more effective use of repeated measurement than does RPA. Ye and Abbe observed that RPA performs close to maximum-likelihood decoding on several Reed-Muller codes. We show
that RPA performance on singular quadratic forms can be improved further by a kernel-aware variant that aggregates projection directions in the same coset of the radical
\(\ker M\), thereby pooling repeated measurements of the same linear component. The two versions agree at full rank, while the kernel-aware decoder improves sharply as the nullity increases. When \(r<\frac{m-1}{2}\), however, CHIRRUP performs better on \(DG(m,r)\). In particular, a matrix in \(DG(m,0)\) is determined by any one of its rows,
and CHIRRUP's tree search combines multiple row estimates more effectively. Essentially, CHIRRUP exploits multiple measurements on $DG(m,0)$ more effectively than vanilla RPA.\\
\hfill\break
{\small\noindent\textbf{Index Terms:}
Reed-Muller codes, AWGN channels, RPA decoding, CHIRRUP decoding,
Delsarte-Goethals sets.\newpage}
\section{Introduction}
In 1942, before the field of coding theory was created, the statistician Fisher \cite{Fi42} discovered the first-order Reed-Muller code $RM(1,m)$ in the course of his work on factorial designs (for more details, see \cite{Ca98}). Codewords in the first-order Reed-Muller code $RM(1,m)$ are simply the rows of the $2^m \times 2^m$ Walsh-Hadamard matrix and their negatives, and decoding can be accomplished very efficiently using the fast Walsh-Hadamard transform.\\
\hfill\break
The discovery of the Reed-Muller (RM) hierarchy by Muller \cite{muller1954application} and Reed \cite{reed1954class} came in 1954. A codeword in the binary $RM$ code $RM(r,m)$ is the evaluation vector of a polynomial $f$ of degree at most $r$ in $m$ binary variables $x_1,x_2,...,x_m$. Reed gave the first decoding algorithm based on majority logic and the recursive structure of Boolean polynomials. Subsequently, Dumer developed recursive list-decoding algorithms that take advantage of the recursive structure of Boolean polynomials to systematically reduce the decoding problem to smaller RM codes (see \cite{dumer2004recursive}, \cite{dumer2006soft}, \cite{dumer2006recursiveLists}). The RPA and CHIRRUP algorithms studied in this paper do apply to higher-order $RM$ codes, but in this paper we restrict our attention to second-order $RM$ codes, in part because their performance on $RM(2,m)$ is superior to their performance on higher-order $RM$ codes.\\
\hfill\break
We consider the problem of recovering the evaluation vector of a noisy quadratic form $Q_M$ associated with a binary skew-symmetric $(m+1)\times(m+1)$ matrix $M$. The defining property of $Q_M$ is that 
\begin{equation}
    Q_M(x+b) = Q_M(x) + Q_M(b) + x^TMb.
\end{equation}
The vector $b$ defines the projection, the vector $d_b = Mb$ defines the direction of the projection, and $x^TMb + Q_M(b)$ is a codeword in the first-order $RM$ code. The Walsh-Hadamard transform can be used to identify $d_b$. Sakkour \cite{sakkour2005decoding} observed that the correspondence between $b$ and $d_b$ is linear, so that given any vector $b' \neq b$, the property $d_b=d_{b'}+d_{b+b'}$ provides an estimate of $d_b$. Sakkour selects $d_b$ by majority vote. The aggregation step in the RPA algorithm \cite{ye2020rpa} is different, using the projections to estimate each indexed coordinate of the evaluation vector.\\
\hfill\break
Section II connects the problem of recovering a noisy quadratic form $Q_M$ associated with a skew-symmetric $(m+1)\times(m+1)$ matrix $M$ with the problem of recovering a noisy $\Z_4$-linear quadratic form $T_P$ associated with a binary symmetric matrix $P$. A $\Z_4$-valued quadratic form $T_P$ is defined to be a map $T_P: \Z_2^m \to \Z_4$ that satisfies
\begin{equation}
    T_P(x+y) = T_P(x)+T_P(y)+2x^TPy,\quad  \textit{for all $x,y\in\Z_2^m$},
\end{equation}
where we view the entries $0$ and $1$ of the binary symmetric matrix $P$ as elements of $\Z_4$. The vector $b$ defines the projection, and $2x^TPb$ is a codeword in the first-order $RM$ code. The Gray map connects the evaluation vectors for $Q_M$ and $T_P$, inducing a correspondence between $M$ and $P$ that is not linear, but is rank-preserving in the sense that if $M$ has rank $m+1-2i$ then $P$ has rank $m+1-2i$ or $m-2i$.\\
\hfill\break
Section III analyzes how the Euclidean distance between evaluation vectors is determined by the rank of the difference between the corresponding matrices $M$ and $P$. Section IV describes the RPA and CHIRRUP decoding algorithms.\\
\hfill\break
Section~V presents numerical experiments involving the Delsarte-Goethals sets
\[
DG(m,r), \qquad r=0,1,\ldots,\frac{m-1}{2},
\]
which are sets of binary symmetric \(m\times m\) matrices such that the
difference of any two distinct matrices has rank at least $m-2r$. These sets provide a hierarchy of decoding problems with different rank distributions and geometry. We use this
hierarchy to compare how RPA and CHIRRUP exploit the structure of the underlying code ensemble.\\
\hfill\break
Both decoders are restricted to the same Delsarte-Goethals set in the aggregation stage and benefit from the same code geometry. However, our simulations show that CHIRRUP was able to use repeated measurements effectively. To better understand the gap between RPA and maximum-likelihood decoding, we introduce a kernel-aware variant of RPA. Given \(\ker M\), this decoder groups projection directions lying in the same coset of \(\ker M\), and therefore pools repeated measurements of the same linear component. The vanilla and kernel-aware decoders agree when \(M\) has full rank, while the kernel-aware variant improves sharply as the nullity increases. This indicates that vanilla RPA does not always fully exploit the repetition created by a
nontrivial kernel in the lower-level Delsarte-Goethals codes.\\
\hfill\break
When $r<\frac{m-1}{2}$, our simulations at $m=7$ show that DG-restricted CHIRRUP outperforms DG-restricted RPA. CHIRRUP reconstructs a common matrix from several autocorrelation measurements, and a matrix in $DG(m,r)$ is determined by any $r+1$ rows. Its tree search can therefore explore multiple row measurements that constrain the same matrix. Across the code ensembles, there is less variation among the RPA performance curves than the CHIRRUP curves. CHIRRUP and RPA assemble evidence that identifies the matrix $M$ in the same way, but CHIRRUP makes more effective use of this evidence at the aggregation stage. \newpage
\section{The Gray Map}
This section fixes the notation for the two representations of second-order Reed-Muller codewords used throughout the paper.  The RPA decoder is naturally described using binary evaluation vectors, while CHIRRUP is naturally described using complex phase vectors arising from $\mathbb Z_4$-valued quadratic forms.  The Gray map is the bridge between these two representations.  It sends a length \(2^m\) quaternary word to a length \(2^{m+1}\) binary word, so that the complex representation on \(\mathbb Z_2^m\) corresponds to the binary Reed-Muller representation on \(\mathbb Z_2^{m+1}\).\\
\hfill\break
The Gray map, shown below, assigns two information bits to the four possible phases in Quadrature Phase-Shift Keying (QPSK) in such a way that the labels of adjacent phases differ by only one binary digit. If the phases are subject to additive white Gaussian noise (AWGN), then the error most likely to occur causes only a single erroneously decoded information bit.\\

\begin{figure}[h]
\centering
\begin{tikzpicture}[scale=0.8]
    % Circle radius
    \def\r{1.5}

    % Draw circle
    \draw[thick] (0,0) circle (\r);

    % Constellation points
    \filldraw[black] (0:\r) circle (2pt);
    \filldraw[black] (90:\r) circle (2pt);
    \filldraw[black] (180:\r) circle (2pt);
    \filldraw[black] (270:\r) circle (2pt);

    % Top point: i, label 1/01
    \node[above=3pt] at (90:\r) {$1/01$};
    \node[below=5pt] at (90:\r) {$i$};

    % Right point: 1, label 0/00
    \node[left=5pt] at (0:\r) {$1$};
    \node[right=5pt] at (0:\r) {$0/00$};

    % Left point: -1, label 2/11
    \node[left=5pt] at (180:\r) {$2/11$};
    \node[right=5pt] at (180:\r) {$-1$};

    % Bottom point: -i, label 3/10
    \node[above=5pt] at (270:\r) {$-i$};
    \node[below=3pt] at (270:\r) {$3/10$};
\end{tikzpicture}
\caption{Gray labeling of QPSK phases.}
\label{fig:qpsk_gray_circle}
\end{figure}
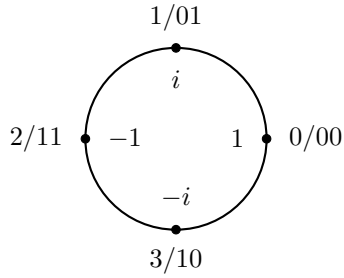

\noindent Distance around the circle defines the Lee metric on $\mathbb{Z}_4$ and the Gray map defines an isometry between $(\mathbb{Z}_4^N, \text{ Lee distance})$ and $(\mathbb{Z}_2^{2N}, \text{ Hamming distance})$ \cite{hammons1994z4}.\\
\hfill\break
We use the Gray map to pass between binary codes and codes defined over the alphabet $\mathbb{Z}_4$. The 2-adic expansion of $c \in \mathbb{Z}_4$ defines two maps $\alpha$ and $\beta$ between $\mathbb{Z}_4$ and $\mathbb{Z}_2$ given by $c = \alpha(c) + 2\beta(c)$ and we extend these maps to vectors, entry by entry. The Gray map $\phi$ is the correspondence
\begin{equation}
\phi(c) = (\beta(c), \alpha(c) + \beta(c)), \quad \text{where } c \in \mathbb{Z}_4^N.
\end{equation}
\hfill\break
\textbf{The Binary Simplex Code $S_m$:} Every linear function $f: \mathbb{Z}_2^m \to \mathbb{Z}_2$ takes the form $f_w(v) = w^Tv$ for some $w \in \mathbb{Z}_2^{m}$, and the codewords in $S_m$ are obtained by evaluating the functions $f_w(v)$ at all vectors $v \in \mathbb{Z}_2^m$. The simplex code $S_m$ is a binary linear code with dimension $m$ and length $2^m$.\\
\hfill\break
\noindent{\bfseries The Binary Reed-Muller Code $RM(1,m)$:} A codeword in the Reed-Muller code $RM(1, m)$ is either a simplex codeword $(f_w(v))$, or the complement $(f_w(v)+1)$, of a simplex codeword. The Reed-Muller code $RM(1, m)$ is a binary linear code with dimension $m+1$ and length $2^m$ \cite{muller1954application,reed1954class}.\\
\hfill\break
\textbf{The Gray Image of the Binary Reed-Muller Code $RM(1, m+1)$:} The binary Reed-Muller hierarchy is defined by the $(u \mid u + v)$ construction  \cite{macwilliams1977theory} and $RM(1, m+1)$ is defined recursively by
\begin{equation}
    RM(1, m+1) = \{ (u \mid u + a\mathbf{1}) \mid u \in RM(1, m), a\in\Z_2\},
\end{equation}
where $\mathbf{1}$ is the binary vector with every entry equal to $1$.\\
The pre-image of $RM(1, m+1)$ under the Gray map is the set
\begin{equation}
   \{ 2w + \lambda\mathbf{1} \mid w \in S_m \text{ and } \lambda \in \mathbb{Z}_4 \},
\end{equation}
where $\mathbf{1}$ is the vector with every entry equal to $1$ \cite[Theorem~8.3]{calderbank1997z4}. In this example, $w$ is a binary vector, but we treat the entries as elements of $\mathbb{Z}_4$, and it is clear from the context that arithmetic takes place in $\mathbb{Z}_4$. We abuse notation in the same way throughout this paper.\\
\hfill\break
\textbf{Binary Quadratic Forms:} If $Q_M: \mathbb{Z}_2^{m+1} \to \mathbb{Z}_2$ is a quadratic form for which the associated bilinear form is $x^T M y$, then by definition
\begin{equation}
Q_M(x + y) = Q_M(x) + Q_M(y) + x^T M y, \quad \text{for all } x, y \in \mathbb{Z}_2^{m+1}.
\end{equation}
Note that the matrix $M$ is skew symmetric, that is, the main diagonal $d(M)=0$. If two quadratic forms are associated with the same bilinear form, then they differ by a linear map $f_w: \mathbb{Z}_2^{m+1} \to \mathbb{Z}_2$. We can represent a quadratic form $Q_M$ by the evaluation vector $(Q_M(v))$ of all values taken on $\mathbb{Z}_2^{m+1}$. If two quadratic forms are associated with the same skew-symmetric matrix $M$ then the corresponding evaluation vectors lie in the same coset of $RM(1, m+1)$ \cite{calderbank1997z4}.\\
\hfill\break
Let $N_0(m+1,r)$ be the number of $(m+1)\times(m+1)$ binary skew-symmetric matrices with rank $r$. For an elementary proof that $r=2s$ is even and that
\begin{equation}
    N_0(m+1,2s)
=
\prod_{i=1}^{s}
\frac{2^{2i-2}}{2^{2i}-1}
\prod_{i=0}^{2s-1}
\left(2^{m+1-i}-1\right),
\quad
N_0(m+1,2s+1)=0.
\end{equation}
see MacWilliams \cite[Theorem 3]{macwilliams1969orthogonal}.\\
\hfill\break
\textbf{$\mathbb{Z}_4$-Valued Quadratic Forms:} Given an $m \times m$ binary symmetric matrix $P$, Brown \cite{brown1972kervaire} defines a $\mathbb{Z}_4$-valued quadratic form $T_P$ to be a map $T_P: \mathbb{Z}_2^m \to \mathbb{Z}_4$ that satisfies
\begin{equation}
T_P(x + y) = T_P(x) + T_P(y) + 2 x^T P y, \quad \text{for all } x, y \in \mathbb{Z}_2^m,
\end{equation}
where we view the entries $0$ and $1$ of the binary symmetric matrix $P$ as elements of $\mathbb{Z}_4$. Setting $y=x$, we obtain $2T_P(x) = 2x^T P x \pmod 4$. In particular, the representative $T_P(x)=x^T P x$ is a $\mathbb{Z}_4$-valued quadratic form associated with the symmetric matrix $P$. If two $\mathbb{Z}_4$-valued quadratic forms are associated with the same symmetric matrix $P$, then they differ by a linear map $2f_w: \mathbb{Z}_2^m \to \mathbb{Z}_4$. We can represent a $\mathbb{Z}_4$-valued quadratic form $T_P$ by the evaluation vector $(T_P(v))$ of all values taken on $\mathbb{Z}_2^m$. If two $\mathbb{Z}_4$-valued quadratic forms are associated with the same symmetric matrix $P$ then the corresponding evaluation vectors lie in the same coset of the Gray image of $RM(1, m+1)$.\\
\hfill\break
Let $N(m,r)$ be the number of $m\times m$ binary symmetric matrices with rank $r$. For an elementary proof that 
\begin{equation}
    N(m,2s)
=
\prod_{i=1}^{s}
\frac{2^{2i}}{2^{2i}-1}
\cdot
\prod_{i=0}^{2s-1}
\left(2^{m-i}-1\right),
\qquad
2s \le m.
\end{equation}
\begin{equation}
    N(m,2s+1)
=
\prod_{i=1}^{s}
\frac{2^{2i}}{2^{2i}-1}
\cdot
\prod_{i=0}^{2s}
\left(2^{m-i}-1\right),
\qquad
2s+1 \le m
\end{equation}
see MacWilliams \cite[Theorem 2]{macwilliams1969orthogonal}.\\
\hfill\break
The Gray map defines a correspondence between quadratic forms and $\mathbb{Z}_4$-valued quadratic forms. Calderbank et al. \cite[Lemma~7.3]{calderbank1997z4} show this correspondence can be expressed as a correspondence between binary symmetric $m \times m$ matrices $P$ and binary skew-symmetric $(m+1) \times (m+1)$ matrices $M$ that is given by
\begin{equation}\label{eq:PtoM} 
  M =
  \begin{pmatrix}
    P + d(P)^{\!T}\,d(P) & d(P)^{\!T} \\[4pt]
    d(P) & 0
  \end{pmatrix},
\end{equation}
where $d(P)$ denotes the main diagonal of $P$. The rank of a skew-symmetric matrix is always even. The correspondence between $M$ and $P$ is not linear, but it is rank-preserving in the sense that if $M$ has rank $m+1-2i$ then $P$ has rank $m+1-2i$ or $m-2i$.\\
\hfill\break
\textbf{Example:} The Gray labeling of QPSK phases converts the evaluation vector of the $\mathbb{Z}_4$-valued quadratic form
\begin{equation}
T_P(x) = x^T P x, \quad \text{with } P=
\begin{bmatrix}
1 & 1 & 0 \\
1 & 0 & 1 \\
0 & 1 & 0
\end{bmatrix},
\end{equation}
to an 8-dimensional codeword in complex Euclidean space with entries $1$, $-1$, $i$, and $-i$. We identify this codeword with the $\mathbb{Z}_4$-valued quadratic form $T_P$. The corresponding binary skew-symmetric matrix is
\begin{equation}
M=
\begin{bmatrix}
0 & 1 & 0 & 1 \\
1 & 0 & 1 & 0 \\
0 & 1 & 0 & 0 \\
1 & 0 & 0 & 0
\end{bmatrix},
\end{equation}
where the variables are ordered as $(x_1,x_2,x_3,u)$. This matrix represents the associated alternating bilinear form. To obtain a binary quadratic form associated with $M$, we choose the upper-triangular representative
\begin{equation}
Q_M(x_1,x_2,x_3,u)=x_1x_2+x_2x_3+x_1u.
\end{equation}
We view the evaluation vector of $Q_M$ as a codeword in $16$-dimensional real Euclidean space with entries $0$ and $1$, and we identify this codeword with the quadratic form $Q_M$.
For example, if \begin{equation}
    x = (x_1,x_2,x_3) = (1,1,0),
\end{equation}
then $T_P(x) = 3 \pmod{4}$, and the corresponding QPSK phase is $i^3 = -i$.\\
The Gray map associates $3$ with the binary pair $10$. We produce the binary pair $10$ by evaluating the quadratic form $Q_M$ twice, once with $u=0$ and once with $u=1$.
\begin{equation}
    Q_M(1,1,0,0) = 1 \pmod{2},\quad
    Q_M(1,1,0,1) = 0 \pmod{2}.
\end{equation}
The discussion above gives the map between the two representations used in this paper: binary quadratic forms in $RM(2,m+1)$ and $\Z_4$-valued quadratic forms on $\Z_2^m$. In the next section we restrict this space to structured families of quadratic forms. The Kerdock and Delsarte-Goethals constructions choose special collections of symmetric matrices whose pairwise differences have controlled rank. Through the Gray map correspondence, these families become structured subcodes of $RM(2,m+1)$. This rank structure is the main feature that will later distinguish the performance of RPA and CHIRRUP.\\
\hfill\break

\section{Kerdock Codes and Delsarte-Goethals Codes}
The RPA algorithm recovers a quadratic form from a noisy observation of its evaluation vector in real Euclidean space. The geometry of this decoding problem is determined by the set of pairwise Euclidean distances between evaluation vectors. Cameron and Seidel \cite{cameron1973quadratic} showed that the inner product $\langle a, b\rangle$ between two evaluation vectors satisfies
\begin{equation}
    |\langle a, b\rangle|^2 = 0 \text{ or } 2^{2(m+1)-k},
\end{equation}
where $k$ is the rank of the difference between the corresponding binary quadratic forms. Hence, the geometry of the problem of recovering a noisy binary quadratic form from a set $S$ is determined by the set of ranks of pairwise differences of quadratic forms taken from $S$. A lower bound on the rank of these pairwise differences translates to a lower bound on the separation of evaluation vectors in real Euclidean space.\\
\hfill\break
\textbf{Structured Subcodes of $RM(2,m+1)$:}
The full second-order Reed-Muller code contains evaluation vectors of all binary quadratic forms. We now introduce two structured subcodes, the Kerdock code and the Delsarte-Goethals codes, by restricting the matrices that define the
quadratic forms. These restrictions control the ranks of pairwise differences of quadratic forms, and therefore control the Euclidean distance between the corresponding evaluation vectors. We will analyze the behavior of these subcodes under projection after introducing the construction.\\
\hfill\break
\textbf{The Nonlinear Binary Kerdock Code:} 
For odd $m$, Kerdock \cite{kerdock1972lowrate} found a set of $2^{m}$ binary quadratic forms with the property that all pairwise differences are non-singular. It is not possible to find a larger set with this property since the first rows of the corresponding skew-symmetric matrices must be distinct. Each quadratic form determines a coset of $RM(1, m+1)$, and the Kerdock code is the union of these cosets. The Kerdock code has length $2^{m+1}$, it contains $2^{2m+2}$ codewords, and the minimum Hamming distance between codewords is $2^m - 2^{(m-1)/2}$. When we replace the entries $0$ and $1$ by $1$ and $-1$ respectively, and we view the resulting vectors in real Euclidean space, we obtain an extremal set of Euclidean lines (see \cite{calderbank1997z4} for details).\\
\hfill\break
The CHIRRUP algorithm recovers a $\Z_4$-linear quadratic form from a set of $\Z_4$-linear quadratic forms on $\Z_2^m$ by finding the closest evaluation vector in complex Euclidean space. Again, the geometry of the decoding problem is determined by the set of pairwise Euclidean distances between evaluation vectors. Given two different $\Z_4$-linear quadratic forms, let $Q$ be the difference between the corresponding symmetric matrices. Calderbank, Howard and Jafarpour \cite[Appendix A]{calderbank2010deterministic} showed that the inner product $S$ between the corresponding evaluation vectors satisfies $|S|^2 = 0$ or $2^{2m-k}$, where $k$ is the rank of $Q$. Again, a lower bound on the rank of pairwise differences translates to a lower bound on the separation of evaluation vectors in complex Euclidean space.\\
\hfill\break
\textbf{The $\mathbb{Z}_4$-Linear Kerdock Code:} This is the Gray image of the nonlinear binary Kerdock code, and Hammons et al. \cite{hammons1994z4} showed that its preimage is a linear code over $\mathbb{Z}_4$. The $2^m$ symmetric matrices corresponding to the constituent $\mathbb{Z}_4$-linear quadratic forms are drawn from an $m$-dimensional binary vector space $DG(m, 0)$, which is the first space in the Delsarte-Goethals hierarchy described below. The nonzero matrices in $DG(m, 0)$ are nonsingular, and the difference of any two distinct matrices in $DG(m,0)$ is also nonsingular. It is not possible to find a larger set with this property since every binary vector of length $m$ appears as the first row of some matrix in $DG(m, 0)$. When we replace the entries in $\mathbb{Z}_4$ with QPSK phases and we view the resulting vectors in complex Euclidean space, we obtain an extremal set of complex Euclidean lines (see \cite{calderbank1997z4} for details).\\
\hfill\break
\textbf{The Delsarte-Goethals Hierarchy:} Let $m$ be odd. The Delsarte-Goethals set $DG(m, r)$ is a binary vector space containing $2^{(r+1)m}$ binary symmetric matrices with the property that the difference of any two distinct matrices has rank at least $k=m-2r$ (see \cite{delsarte1975alternating}, \cite[Chapter 15]{macwilliams1977theory}, \cite{hammons1994z4}). The Delsarte--Goethals sets are nested 
\begin{equation}
    DG(m,0) \subset DG(m,1) \subset \cdots \subset DG\left(m,\frac{m-1}{2}\right).
\end{equation}
The first set $DG(m, 0)$ determines the $\mathbb{Z}_4$-linear Kerdock code, and the nonlinear binary Kerdock code via the Gray map. The last set $DG(m, (m-1)/2)$ is the set of all binary symmetric matrices which determine the second-order Reed-Muller code $RM(2, m+1)$. We now construct the sets $DG(m, r)$ using the trace map in finite fields.\\
\hfill\break
Let $N=2^m$ and let $\F_N$ be the finite field obtained by adjoining a root $\alpha$ of an irreducible polynomial of degree $m$ over the binary field $\Z_2$. The field $
F_N$ is an $m$-dimensional binary vector space with basis $1, \alpha, \alpha^2, \ldots, \alpha^{m-1}$, and we represent the field element $
x = \sum_{i=0}^{m-1} c_i \alpha^i$ by the binary vector $x = (c_0, \ldots, c_{m-1})^T$. The trace map $\operatorname{Tr}: F_N \to Z_2$ is given by
\begin{equation}
    \operatorname{Tr}(x) = x + x^2 + x^4 + \ldots + x^{N/2}.
\end{equation}
The automorphisms of the field $\F_N$ that fix the binary field $Z_2$ are generated by the squaring operator $F: x \to x^2$, which we view as a nonsingular $m \times m$ binary matrix. The trace is a linear map, $\operatorname{Tr}(xy)$ is a symmetric bilinear form, and there exists a nonsingular binary matrix $P$ for which $\operatorname{Tr}(xy) = x^T Py$.
For each $z \in \F_N$, let $A_z$ denote the $m \times m$ binary matrix
representing the $\mathbb{Z}_2$-linear map
\[
\mu_z : \F_N \longrightarrow \F_N,
\qquad
\mu_z(x)=xz.
\]
The Delsarte--Goethals set $DG(m,r)$ is the vector space of all binary
symmetric matrices of the form
\begin{equation}
Q_{\mathbf z}
=
A_{z_0}P+
\sum_{i=1}^{r}
\left(
A_{z_i}P(F^i)^{T}
+
F^iPA_{z_i}^{T}
\right),
\qquad
\mathbf z=(z_0,z_1,\ldots,z_r)\in\mathbb F_N^{r+1}.
\end{equation}

Suppose that $\mathbf z\neq 0$ and that $x\in\ker Q_{\mathbf z}$. Then,
for every $y\in\mathbb F_N$,
\begin{align}
0
&=
\operatorname{Tr}\left(
z_0xy+
\sum_{i=1}^{r}
z_i\left(x^{2^i}y+xy^{2^i}\right)
\right) \\
&=
\operatorname{Tr}\left(
y^{2^r}
\left[
z_0^{2^r}x^{2^r}
+
\sum_{i=1}^{r}
\left(
z_i^{2^r}x^{2^{r+i}}
+
z_i^{2^{r-i}}x^{2^{r-i}}
\right)
\right]
\right).
\end{align}
Since the map $y\mapsto y^{2^r}$ is a bijection of $\mathbb F_N$ and
the trace pairing is nondegenerate, $x$ is a root of the linearized
polynomial
\begin{equation}
z_0^{2^r}x^{2^r}
+
\sum_{i=1}^{r}
\left(
z_i^{2^r}x^{2^{r+i}}
+
z_i^{2^{r-i}}x^{2^{r-i}}
\right)
=0.
\end{equation}
This is a nonzero polynomial of degree at most $2^{2r}$, so it has at
most $2^{2r}$ roots. Therefore,
\[
\dim\ker Q_{\mathbf z}\le 2r,
\qquad
\operatorname{rank}(Q_{\mathbf z})\ge m-2r.
\]

\hfill\break
\textbf{First-Order Projections of Quadratic Forms: }When a quadratic form is projected in one direction, the resulting object is first order. This structure is the algebraic feature exploited by both RPA and CHIRRUP. The following theorem is proved in Appendix A.
\begin{theorem}\label{Projection}[Binary first-order projection]
Let $V=\Z_2^{m+1}$, and let $Q_M:V\to \Z_2$ be a binary quadratic form whose
associated bilinear form is $x^TMy$, where $M$ is a binary skew-symmetric
$(m+1)\times(m+1)$ matrix. For $b\in V$, define
\begin{equation}
D_bQ_M(x)=Q_M(x+b)+Q_M(x).
\end{equation}
Then
\begin{equation}
D_bQ_M(x)=x^TMb+Q_M(b).
\end{equation}
In particular, if $b\neq 0$, then $D_bQ_M$ is a first-order Reed-Muller word on
the quotient space $V/\langle b\rangle$.

If $\operatorname{rank}(M)=m+1-2r$, then as $b$ varies over $V$, the possible
linear parts of $D_bQ_M$ form the column space of $M$, and each linear part
occurs exactly $2^{2r}$ times. Equivalently,
\begin{equation}
\Pr_{b\in V}[Mb=0]=2^{2r-(m+1)},
\end{equation}
where $b$ is chosen uniformly from $V$.
\end{theorem}
If $b' = b+k$ for some $k\in\ker M$, then
\[
Mb'=M(b+k)=Mb.
\]

Thus, projection directions that differ by an element of $\ker M$ produce
the same linear component. When
$\operatorname{rank}(M)=m+1-2r$, each linear component is therefore
measured through $2^{2r}=|\ker M|$ projection directions. These duplicate
measurements provide additional information for the aggregation step in RPA,
although the corresponding affine projections may differ in their constant
terms. Section~IV describes how RPA and CHIRRUP use first-order
projections, and Section~V returns to the dependence of their performance
on rank.

\newpage
\section{RPA and CHIRRUP Decoding Algorithms}
The RPA and CHIRRUP algorithms both take advantage of the recursive structure of binary Reed-Muller codes. Every codeword of length $2^{m+1}$ in the $r$-th order Reed-Muller code $\mathrm{RM}(r,m+1)$ can be written as a concatenation $(u \mid u+v)$ where $u$ is a codeword in $RM(r,m)$ and $v$ is a codeword in $RM(r-1,m)$. RPA and CHIRRUP can be applied to Reed-Muller codes of any order and are particularly effective when applied to second-order Reed-Muller codes.\\
\hfill\break
\textbf{RPA Algorithm:} Suppose that $Q_M$ is a quadratic form on $\Z_2^{m+1}$ where the skew-symmetric matrix $M$ defines the associated bilinear form. A nonzero vector $b \in \Z_2^{m+1}$ defines a quotient space $Z_2^{m+1}/\langle b\rangle$. The cosets of $\langle b\rangle$ determine $2^m$ coordinate pairings $(x,x+b)$ as $x$ runs through the quotient space. We have divided the $2^{m+1}$ coordinates into two halves and paired the coordinate $x$ with the coordinate $x+b$. We use the defining property of a quadratic form to write the evaluation vector determined by $Q_M$ as
\begin{equation}
(Q_M(x),Q_M(x+b))=(Q_M(x),Q_M(x))+(0,x^T M b)+(0,Q_M(b)).
\end{equation}
Here all additions are in $\Z_2$. The evaluation vector $\bigl(x^\top M b+Q_M(b)\bigr)$ is a codeword in the first-order Reed-Muller code $RM(1,m)$. RPA selects different projection vectors $b$ in order to query the skew-symmetric matrix $M$ in different ways. RPA decoding involves the following three steps:\\
\hfill\break
\noindent\emph{\textbf{Projection:}} Given a noisy evaluation vector $(v(x))$ and a projection vector $b \in \mathbb{Z}_2^{m+1}$, form the vector $(v(x)+v(x+b))$. The analysis given above shows $(v(x)+v(x+b))$ is a first-order codeword in $RM(1,m)$ corrupted by noise.\\
\hfill\break
\noindent\emph{\textbf{Recursion:}} Decode $(v(x)+v(x+b))$ to a codeword in $RM(1,m)$ using the fast Walsh-Hadamard transform (FWHT). For Reed-Muller codes of order $r$, RPA decodes $(v(x)+v(x+b))$ recursively to a codeword in $RM(r-1,m)$.\\
\hfill\break
\noindent\emph{\textbf{Aggregation:}} For every coordinate of the initial noisy codeword, use a structured majority vote or soft-decision aggregation rule to fuse the information derived from the different projections.\\
\hfill\break
\noindent\textbf{Translation from Real to Complex Euclidean Space:} We start in real Euclidean space with an instance $(\mathbf{x}+\mathbf{n_1},\mathbf{y}+\mathbf{n_2})$ of the decoding problem, where $(\mathbf{x},\mathbf{y})$ is a codeword in $RM(2,m+1)$, and the vectors $\mathbf{n_1},\mathbf{n_2}$ represent additive white Gaussian noise (AWGN). In the codeword $(\mathbf{x},\mathbf{y})$ we have pairs of coordinates $(0,0)$, $(1,0)$, $(0,1)$ and $(1,1)$ with squared Euclidean distances $1$ and $2$. We apply the inverse of the Gray map to $(\mathbf{x},\mathbf{y})$ to obtain a codeword $\mathbf{z}$ in complex Euclidean space with entries $1,-1,i$ and $-i$. Since the squared Euclidean distances between entries are now $2$ and $4$, we normalize $\mathbf{z}$ by $\sqrt{2}$. We replace the real vector $(\mathbf{n_1},\mathbf{n_2})$ by the complex vector $\mathbf{n_1}+i\mathbf{n_2}$ in order to preserve the geometry of the decoding problem. This translates the problem of recovering a noisy quadratic form in real Euclidean space to that of recovering a noisy $\mathbb{Z}_4$-linear quadratic form in complex Euclidean space. Our method of translation preserves the geometry of the decoding problem.\\
\hfill\break
\noindent\textbf{CHIRRUP Algorithm: }Suppose that $T_P$ is a $\Z_4$-linear quadratic form on $\Z_2^m$ where the symmetric matrix $P$ defines the associated bilinear form. The $\Z_4$-linear quadratic form $T_P(x)+2d^Tx+f$ is represented by the evaluation vector 
\begin{equation}
\Psi_{P,d,f}(x)
= \frac{1}{\sqrt{2^m}} i^{\,T_P(x)+2d^Tx+f},
\qquad x\in \mathbb Z_2^m .
\end{equation}
We translate the entries by $b\in\Z_2^{m}$, and set
\begin{equation}
G_b(x)
=
\Psi_{P,d,f}(x+b)
\overline{\Psi_{P,d,f}(x)}
=
2^{-m}
i^{T_P(b)+2b^TPx+2d^Tb}.
\end{equation}
The factor $i^{T_P(b)+2d^Tb}$ is independent of $x$. After
removing this constant phase and rescaling, we obtain
\[
2^m i^{-T_P(b)-2d^Tb}G_b(x)
=
(-1)^{b^TPx}
=
(-1)^{(Pb)^Tx},
\]
where the last equality follows because $P$ is symmetric. Thus,
after removal of the constant phase, $G_b$ is a row of the
Walsh-Hadamard matrix of order $2^m$, scaled by $2^{-m}$.
CHIRRUP decoding involves the following three steps.\\
\hfill\break
\noindent\emph{\textbf{Projection:}} Given a noisy evaluation vector $v(x)$, and a vector $b \in Z_2^m$, form the vector $(v(x+b)\overline{v(x)})$. The analysis given above shows $(v(x+b)\overline{v(x)})$ is obtained by scaling a row of the Walsh-Hadamard matrix of order $2^m$ then adding Gaussian noise.\\
\hfill\break
\noindent\emph{\textbf{Recursion:}} Derive an estimate for $2b^T P x+2d^T b$ by decoding $(v(x)\overline{v(x+b)})$ to a codeword in $RM(1,m)$ using the Fast Hadamard Transform.\\
\hfill\break
\noindent\emph{\textbf{Aggregation:}} Every translate $b$ provides information about the rows of $P$. Use a structured tree search to aggregate the information derived from the different translations (MATLAB code is available to download from\\ \texttt{https://github.com/ajthompson42/CHIRRUP}).\\
\hfill\break
The Delsarte-Goethals set $DG(m,r)$ is an ensemble of $\Z_4$-linear quadratic forms with the property that $r+1$ rows of the matrix $P$ uniquely determine the form. This property simplifies the aggregation step in CHIRRUP.\\
\hfill\break
\noindent\textbf{Recovering Superpositions of Evaluation Vectors:}
After scaling by $2^{-m/2}$, the evaluation vectors corresponding to the same symmetric matrix $P$ determine an orthonormal basis. The $\Z_4$-linear Kerdock code determines a collection of $2^m$ orthonormal bases that are mutually unbiased, in that a vector in one orthonormal basis looks like Gaussian noise to any other basis. This is because, given two different symmetric matrices from the Kerdock set $DG(m,0)$, the inner product of the corresponding scaled evaluation vectors $u_P,u_Q$ satisfies $|\langle u_P,u_Q\rangle|=2^{-m/2}$. Calderbank and Thompson \cite{calderbank2020chirrup} have enhanced the chirp detection algorithm \cite{calderbank2010deterministic} described above to develop a peeling decoder that uses the geometry of orthonormal bases to recover a superposition of evaluation vectors. Their peeling decoder supports massive multiple access in wireless communications (for more details, see \cite{calderbank2020chirrup}). Essentially the same algorithm can be applied to binary codewords from $RM(2,m)$ after exponentiation of entries, that is, after replacing an entry $0$ by $1$ and an entry $1$ by $-1$.\\
\section{Numerical Experiments}
We fix $m=7$, so that codewords in $RM(2,m+1)$ have length $256$, and codewords in the corresponding $\Z_4$-linear code have length $128$. For $r=0,1,2,$ and $3$, we draw a random symmetric matrix $P$ from the Delsarte-Goethals set $DG(7,r)$, and we construct the corresponding skew-symmetric matrix $M$. We construct the evaluation vector $z_P$ of a random $\Z_4$-linear quadratic form $T_P$ corresponding to $P$, then apply the Gray map to obtain the evaluation vector $z_Q$ of the binary quadratic form $Q_M$ corresponding to $M$. We generate a vector $n_P$ of AWGN samples in complex Euclidean space $\mathbb{C}^{128}$, then use the method described in Section IV to generate the paired vector $n_Q$ of AWGN samples in real Euclidean space $\mathbb{R}^{256}$. We parameterize the channel directly by the inverse noise spectral density $1/N_0$. This normalization is independent of the rate of the particular Delsarte-Goethals subcode, so all code ensembles are compared at the same channel noise level. We present the instance $z_P+n_P$ to DG-restricted CHIRRUP and the instance $z_Q+n_Q$ to DG-restricted RPA.\\
\hfill\break
DG-restricted CHIRRUP and DG-restricted RPA use the prescribed Delsarte-Goethals set in the same role: the restriction constrains the reconstruction of the quadratic form, while the two decoders retain their original mechanisms for extracting evidence from the received signal. DG-restricted CHIRRUP reconstructs the quadratic form row by row, retaining at each stage the highest-scoring candidate that is consistent with the prescribed Delsarte-Goethals set. DG-restricted RPA first forms and decodes all 255 nonzero one-dimensional projections exactly as in standard RPA. Its restriction is introduced only in the aggregation stage. Each decoded projection provides a linear label determined by the transmitted quadratic form and the corresponding projection direction. These labels are aggregated to score row-by-row candidates that are consistent with the prescribed Delsarte-Goethals set. At each stage, only the highest-scoring candidate is retained. After the quadratic form has been reconstructed, both decoders recover the remaining affine component by a final Walsh--Hadamard transform. In this way, both decoders have been altered to intake Delsarte-Goethals code ensembles, while the first two stages remain unchanged.\\
\hfill\break\newpage
We calculate the error rate as a function of increasing $1/N_0$, and we follow \cite{ye2020rpa} in presenting results for word error rate (WER) rather than bit error rate (BER). For each of the resulting $44$ $(r,1/N_0)$ cells, our WER is obtained by averaging over $10,000$ paired trials.\\
\hfill\break
The geometry of the CHIRRUP and RPA decoding problems for $DG(m,r)$ is the set of pairwise distances between evaluation vectors. As $r$ increases, the minimum squared distance between evaluation vectors decreases. Fig. 2 shows that both decoders generally improve as $r$ decreases and pairs of evaluation vectors move farther apart. Thus, geometry benefits both decoders, but it does not by itself explain the much stronger separation across r observed for restricted-CHIRRUP than restricted-RPA.
\begin{figure}[H]
    \centering
    \captionsetup{font=small,skip=2pt}

    \includegraphics[width=0.60\textwidth]
    {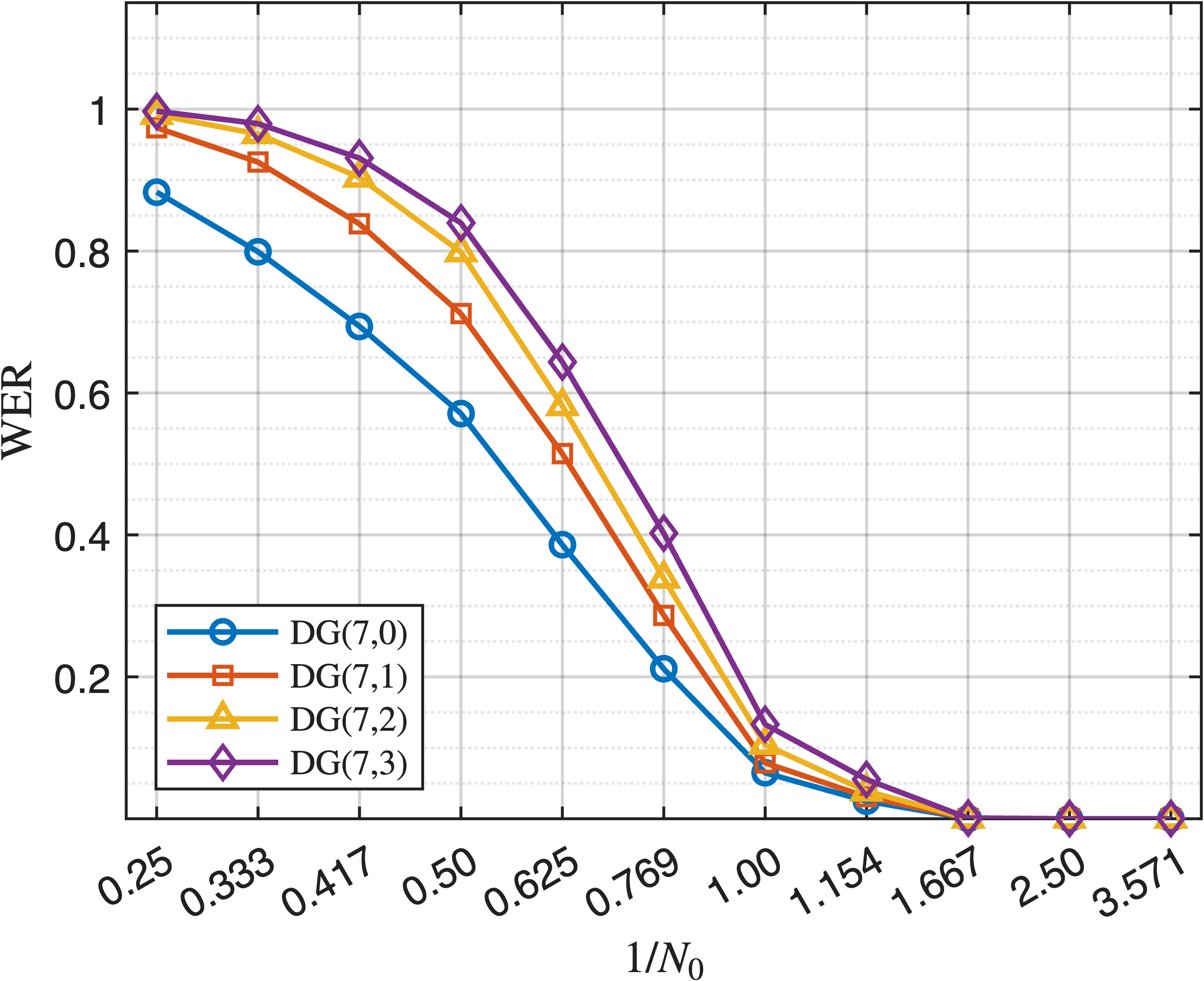}

    \vspace{-0.10cm}
    {\footnotesize Fig.~\ref{fig:method-only}a. DG-Restricted CHIRRUP}

    \vspace{0.10cm}

    \includegraphics[width=0.60\textwidth]
    {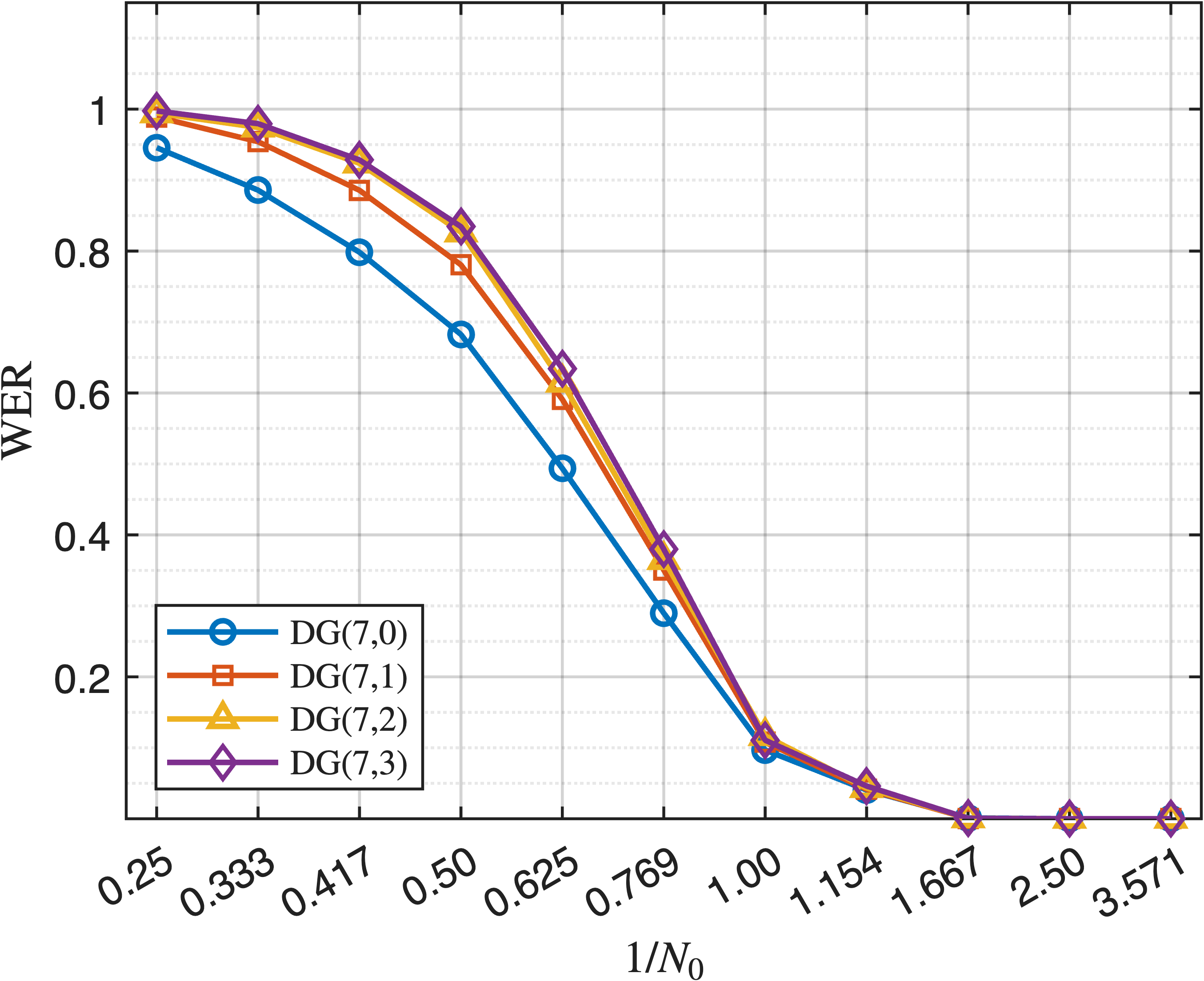}

    \vspace{-0.10cm}
    {\footnotesize Fig.~\ref{fig:method-only}b. DG-Restricted RPA}

    \vspace{0.05cm}

    \caption{DG-Restricted Decoders performance for $DG(m,r)$ as a function of increasing SNR.}
    \label{fig:method-only}
\end{figure}
Fig.~\ref{fig:paired-comparison} compares DG-restricted RPA and DG-restricted CHIRRUP performance on
$DG(7,r)$ for $r=0,1,2,$ and $3$.
\begin{figure}[H]
    \centering
    \begin{minipage}[t]{0.35\textwidth}
        \centering
        \includegraphics[width=\linewidth]{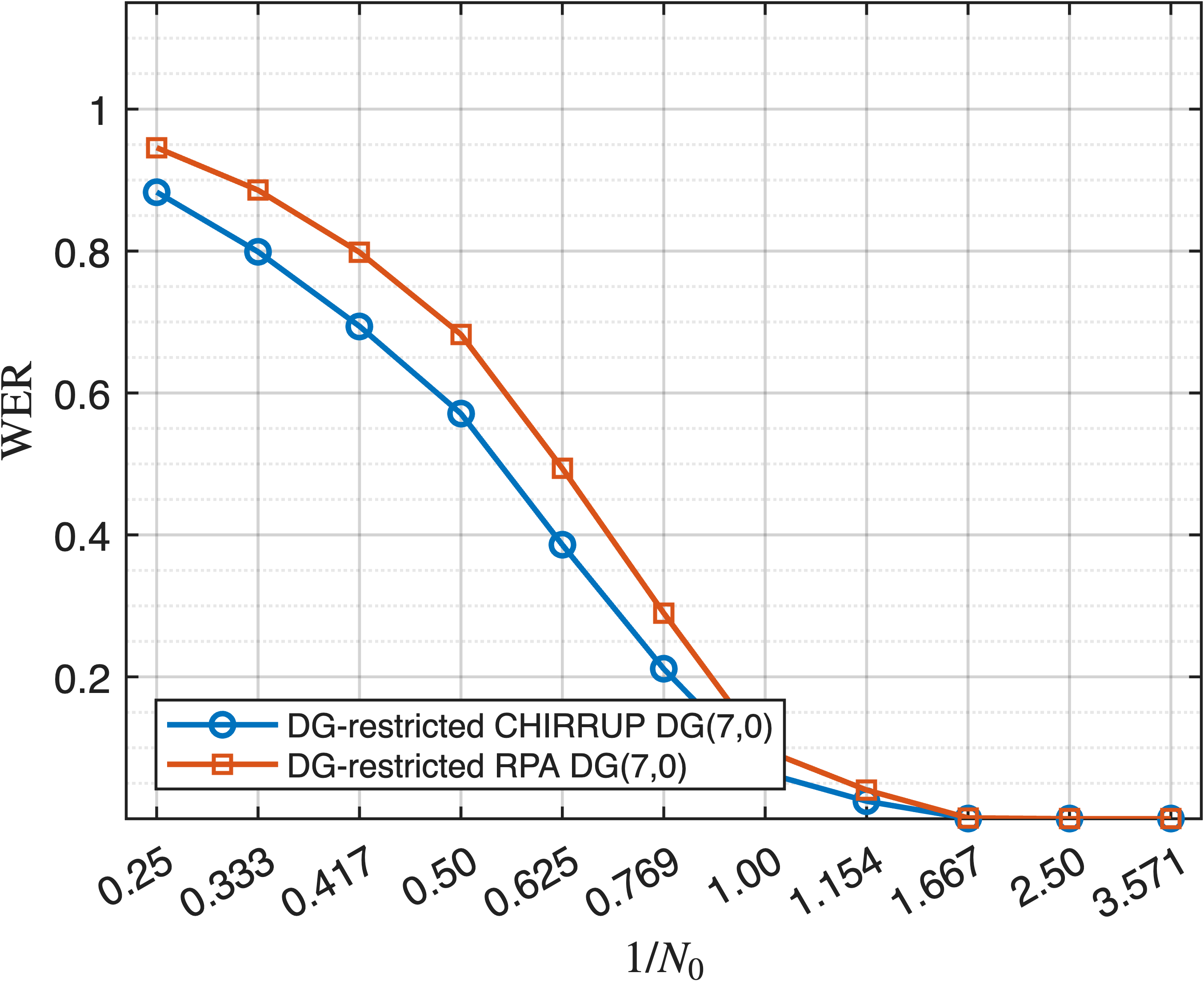}
        \caption*{\footnotesize Fig.~\ref{fig:paired-comparison}a. $DG(7,0)$}
    \end{minipage}
    \hspace{0.02\textwidth}
    \begin{minipage}[t]{0.35\textwidth}
        \centering
        \includegraphics[width=\linewidth]{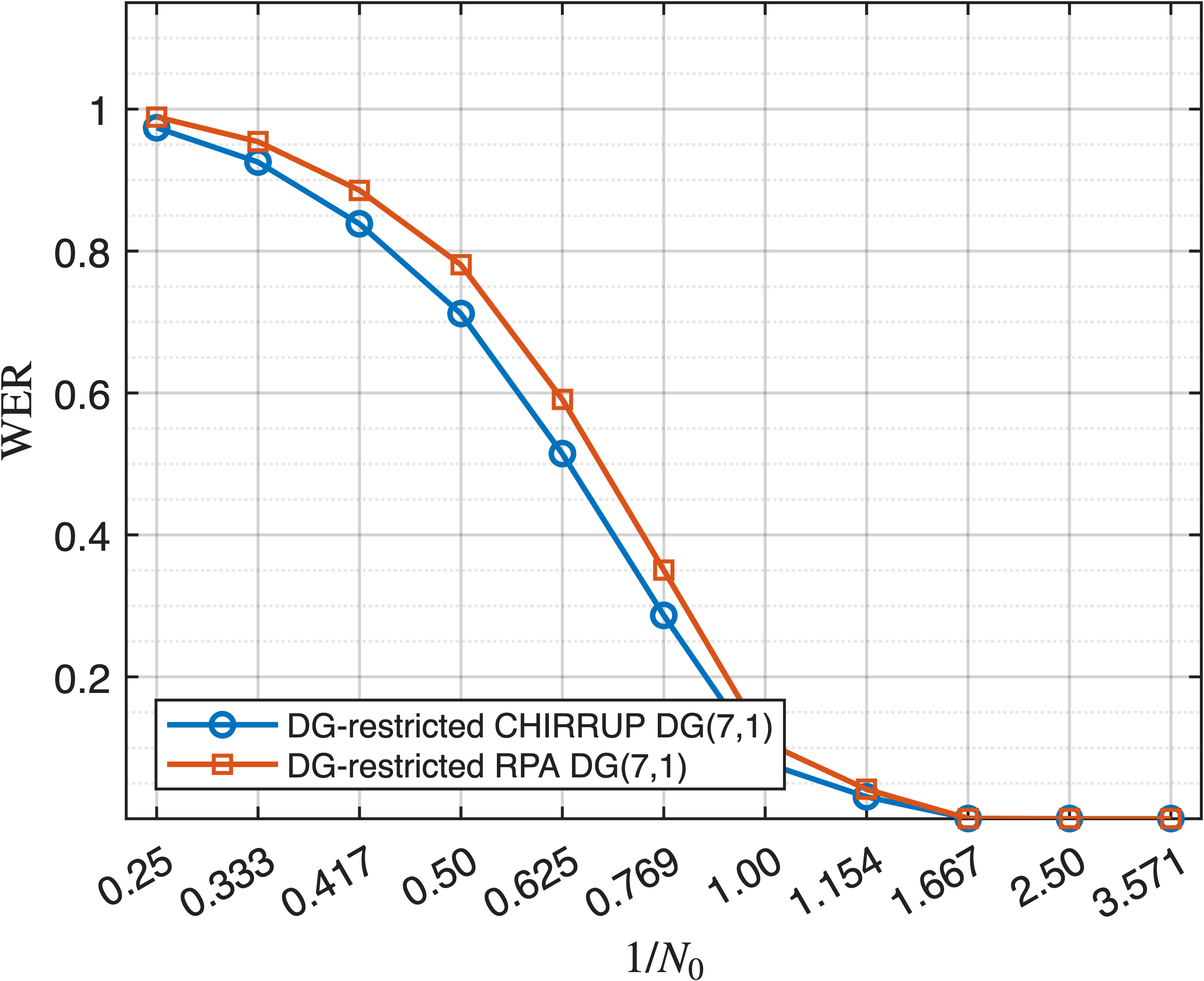}
        \caption*{\footnotesize Fig.~\ref{fig:paired-comparison}b. $DG(7,1)$}
    \end{minipage}

    \vspace{0.15cm}

    \begin{minipage}[t]{0.35\textwidth}
        \centering
        \includegraphics[width=\linewidth]{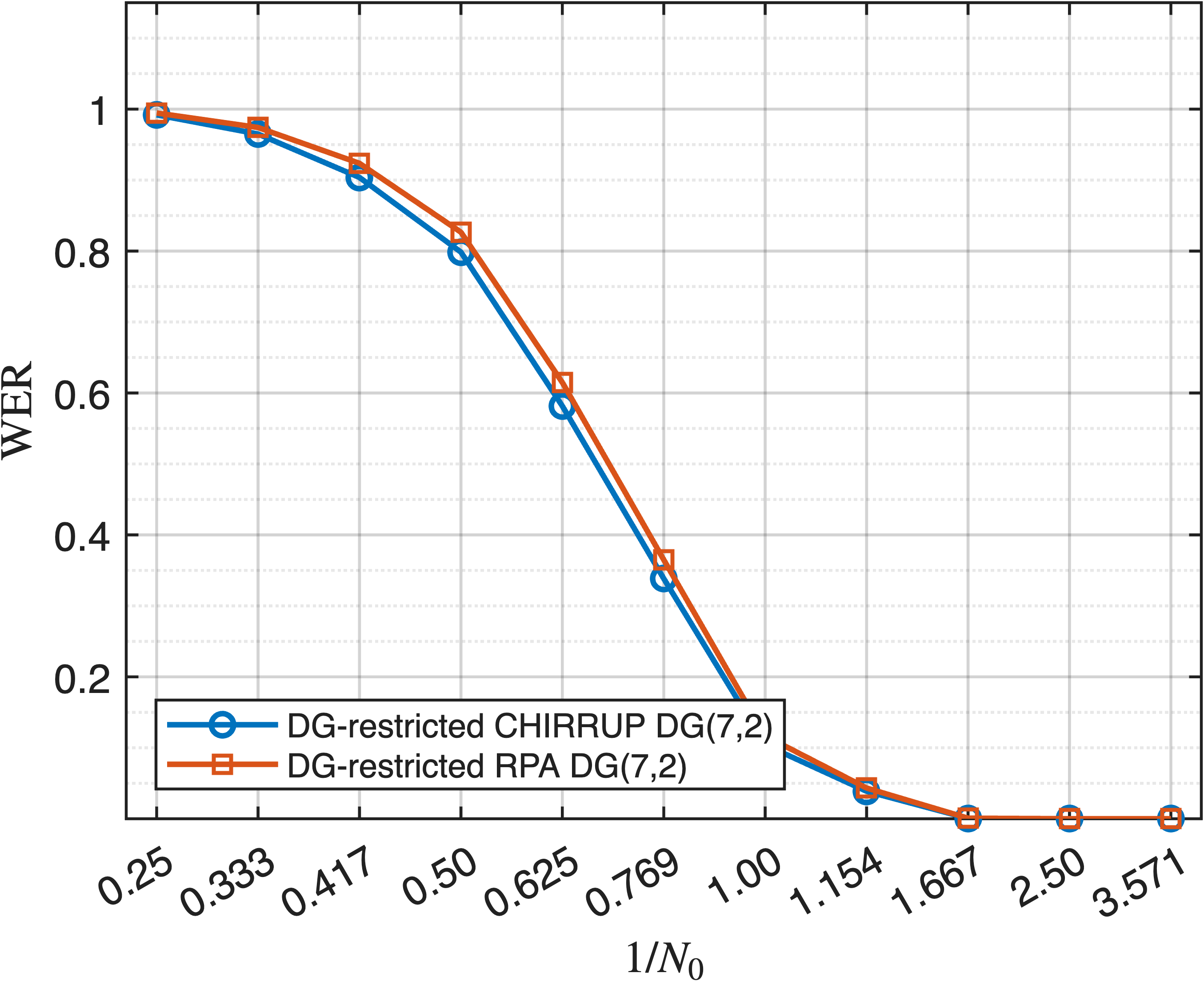}
        \caption*{\footnotesize Fig.~\ref{fig:paired-comparison}c. $DG(7,2)$}
    \end{minipage}
    \hspace{0.02\textwidth}
    \begin{minipage}[t]{0.35\textwidth}
        \centering
        \includegraphics[width=\linewidth]{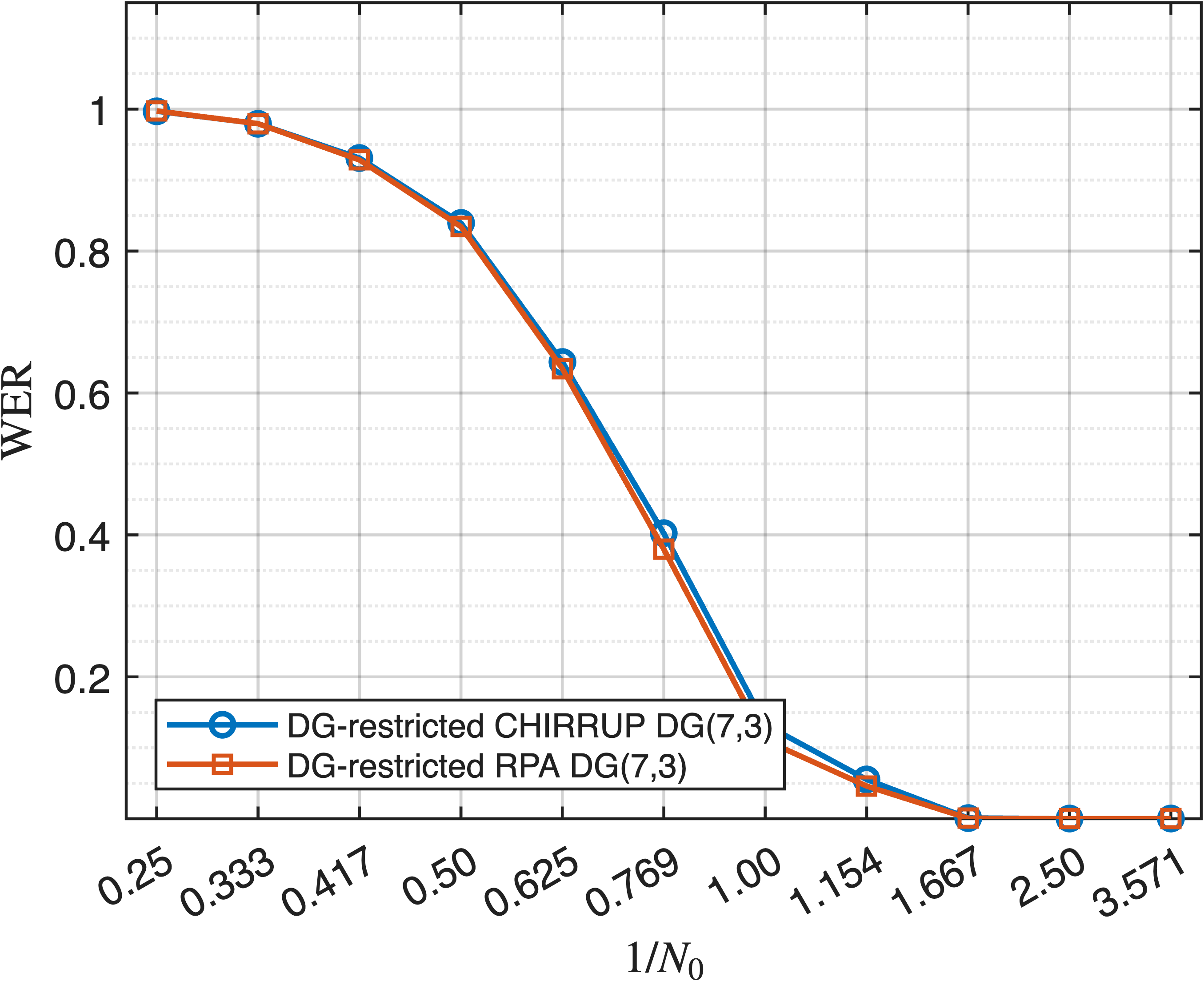}
        \caption*{\footnotesize Fig.~\ref{fig:paired-comparison}d. $DG(7,3)$}
    \end{minipage}

    \caption{Comparison of RPA and CHIRRUP performance as a function of increasing SNR.}
    \label{fig:paired-comparison}
\end{figure}
Figure~4 examines the original decoder failures on \(DG(7,3)\). Conditional on a
word error, we record
\(\operatorname{rank}(P-\widehat P)\) for CHIRRUP and
\(\operatorname{rank}(M-\widehat M)\) for RPA. The latter rank is
necessarily even because \(M-\widehat M\) is alternating.
\hfill\break
\begin{figure}[H]
    \centering

    \begin{minipage}[t]{0.40\textwidth}
        \centering
        \includegraphics[width=\linewidth]
        {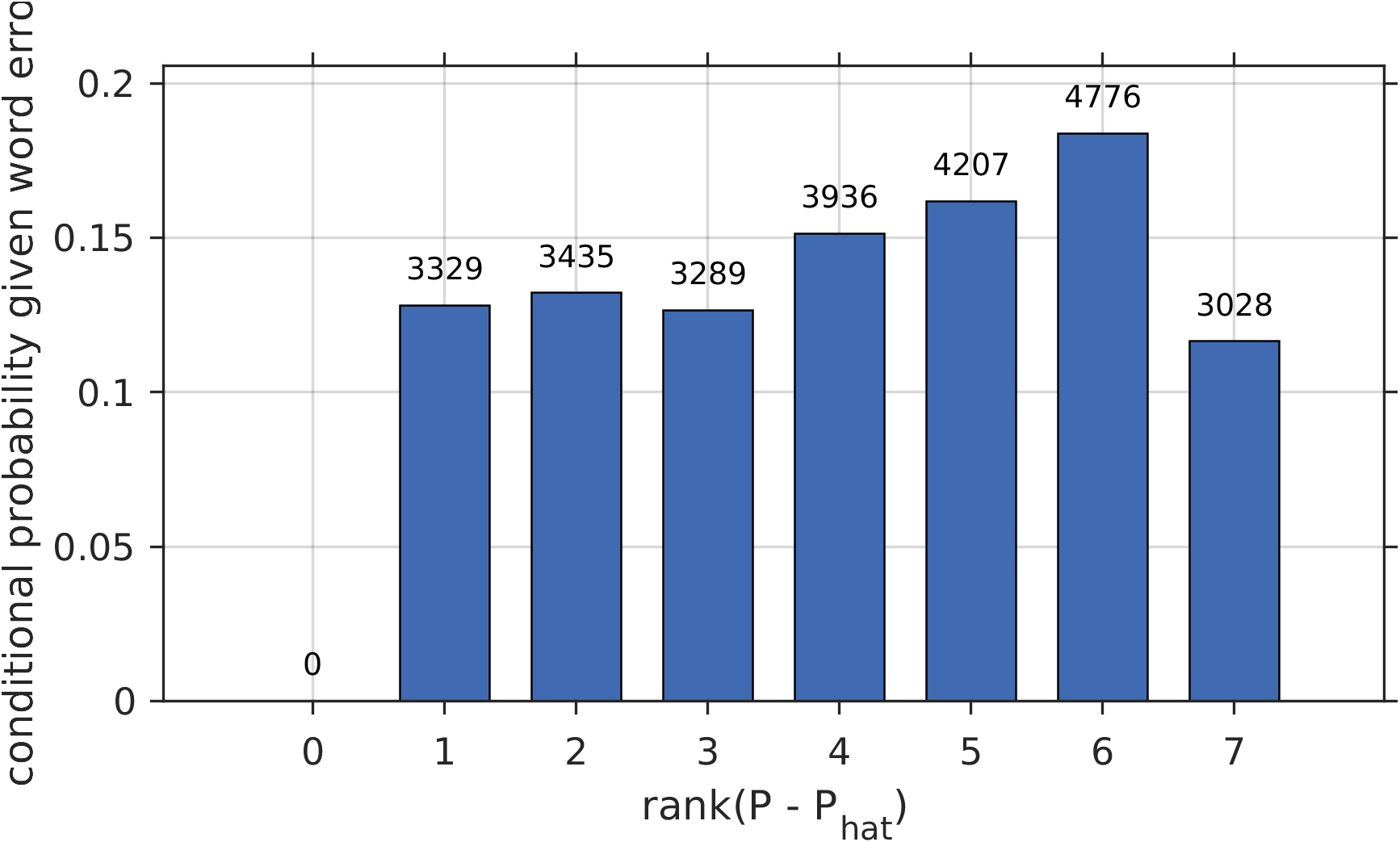}
        \caption*{Fig.~\ref{fig:rank-errors}a. CHIRRUP}
    \end{minipage}
    \hspace{0.04\textwidth}
    \begin{minipage}[t]{0.40\textwidth}
        \centering
        \includegraphics[width=\linewidth]
        {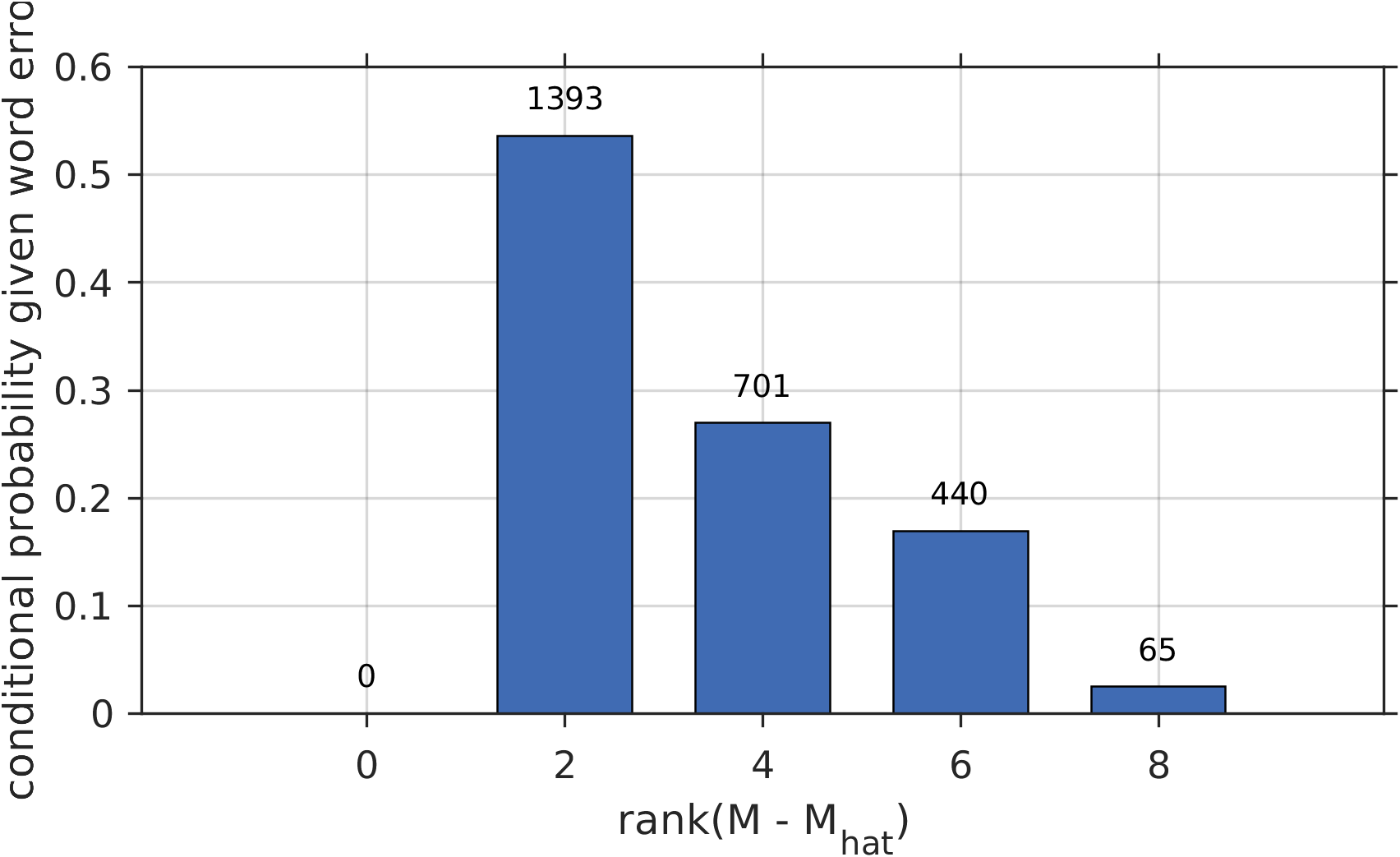}
        \caption*{Fig.~\ref{fig:rank-errors}b. RPA}
    \end{minipage}

    \caption{Conditional distribution of the rank difference between the
    transmitted and incorrectly decoded quadratic forms.}
    \label{fig:rank-errors}
\end{figure}

Figure 4 shows that CHIRRUP and RPA fail in different ways. The CHIRRUP errors are distributed across all nonzero ranks of
$P-\widehat{P}$, while the RPA errors are concentrated at low even ranks of
$M-\widehat{M}$. The rank-enumeration formulas in Section~II show that high-rank matrix differences are much more numerous than low-rank differences.
Nevertheless, low ranks account for a substantial fraction of the observed errors, particularly for RPA. Thus, CHIRRUP errors occur across a broad range of ranks, whereas RPA failures are disproportionately associated with low-rank competitor forms.\\
\hfill\break
The maximum-likelihood (ML) decoder selects the most probable codeword from the ensemble of possible codewords. Since the noise is Gaussian, the ML
decoder selects the codeword that is closest to the received vector. The geometry of the CHIRRUP and
RPA decoding problems is least favorable for
\(DG\bigl(m,\frac{m-1}{2}\bigr)\), because this ensemble contains pairs of
quadratic forms whose difference has rank \(2\). Fig.~2 raises the question of why 
under the same aggregation constructions, CHIRRUP, when restricted on DG code ensembles, performs better than RPA. We show that performance also depends on the availability of multiple measurements, and CHIRRUP is able to gain more information from repeated measurements.\\
\hfill\break
Suppose that the transmitted word is determined by \(M\), and let
\(\widehat{M}\) denote a competing matrix. Write
\[
\Delta M=M-\widehat{M}.
\]
For a projection direction \(b\), RPA obtains information about \(M\) through
the product \(Mb\), while the information distinguishing \(M\) from
\(\widehat{M}\) is contained in \(\Delta M b\). If
\[
\operatorname{rank}(\Delta M)=m+1-2r,
\]
then Theorem~III.1 implies that the map
\[
b\longmapsto \Delta M b
\]
has a kernel of size \(2^{2r}\). Consequently, as \(b\) runs through the
possible directions, each element of the column space of \(\Delta M\) appears
\(2^{2r}\) times as a product \(\Delta M b\).\\
Thus, a low-rank difference also provides RPA with more repeated
measurements of each distinct projected difference. These repeated projected
differences can reinforce the same coordinate estimates during aggregation;
ordinary RPA uses this redundancy only indirectly and does not explicitly group
kernel-equivalent directions. This helps explain both the concentration of RPA
errors at low ranks in Fig.~4 and the strong performance of RPA on the larger
Delsarte--Goethals ensembles. By contrast, in $DG(m,0)$, every nonzero
difference has full rank, so the map
\[
b\longmapsto \Delta M b
\]
is one-to-one and provides no repeated evidence. At the same time,
a matrix in $DG(m,0)$ is determined by any one of its rows, while CHIRRUP's
tree search combines row estimates obtained from multiple translations.
CHIRRUP can therefore exploit several redundant constraints on the same
matrix, which helps explain its advantage over RPA on the smaller ensembles
in Fig.~2.\\
\hfill\break
The preceding analysis concerns repeated projected differences $\Delta M b$
between the transmitted matrix and a competing matrix. To test whether
repeated projection measurements can improve RPA when their equivalence
structure is known, we perform a complementary experiment. We compare
ordinary RPA with a kernel-aware version that is given the true null space
$\ker M$. The modified decoder groups projection directions lying in the same
coset of $\ker M$, which produce the same linear component $Mb$, and combines
the resulting duplicate measurements. The projection and recursion steps
remain the same as in ordinary RPA. Both decoders were evaluated on the same
transmitted words and noise realizations in Figure~5.
\begin{figure}[H]
    \centering
    \captionsetup{font=small,skip=2pt}

    \begin{minipage}[t]{0.33\textwidth}
        \centering
        \includegraphics[width=\linewidth]
        {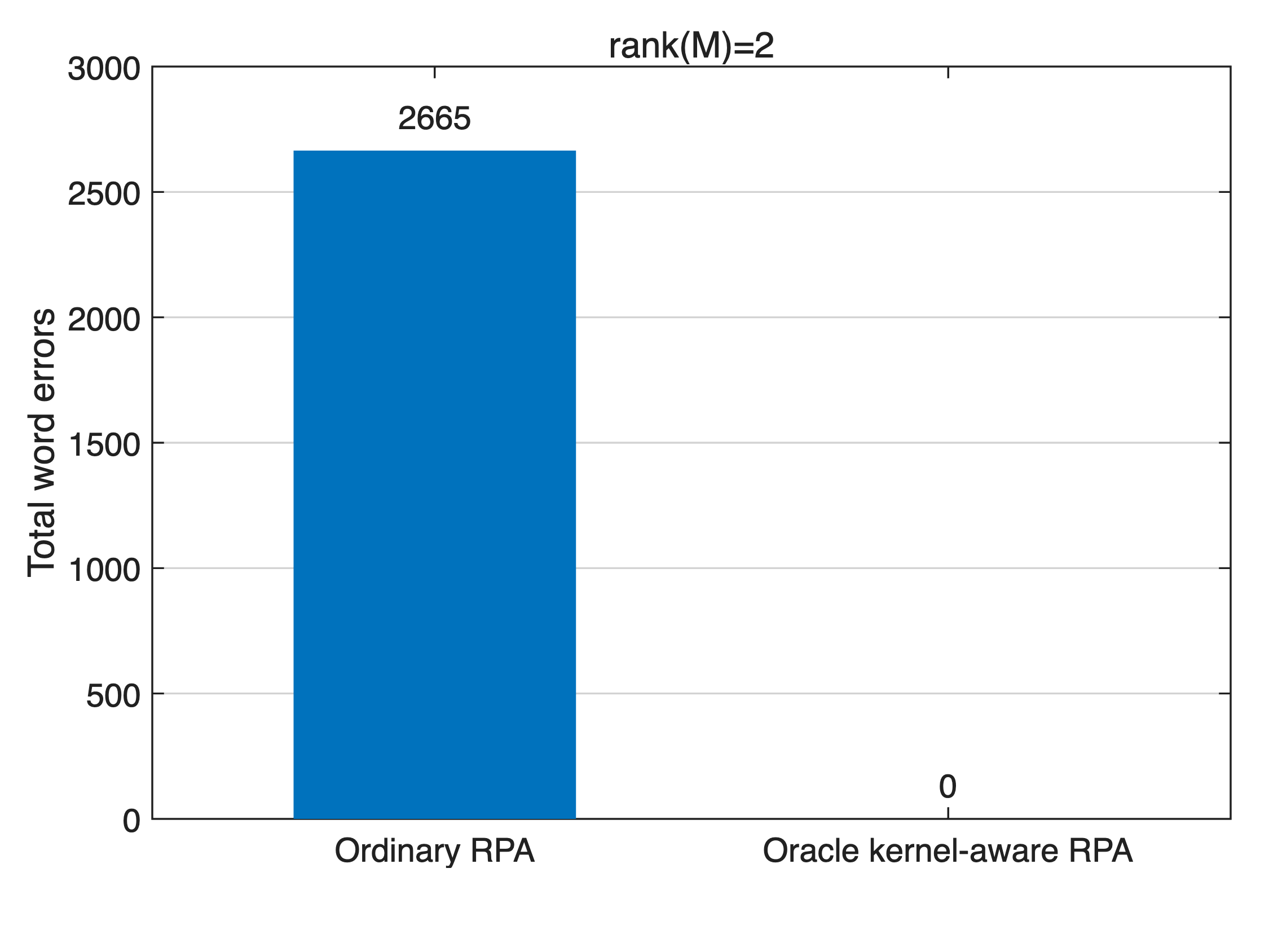}
        \vspace{0.2cm}
        {\scriptsize Fig.~\ref{fig:kernel-aware-rpa}a.
        $\operatorname{rank}(M)=2$}
    \end{minipage}
    \hspace{0.04\textwidth}
    \begin{minipage}[t]{0.33\textwidth}
        \centering
        \includegraphics[width=\linewidth]
        {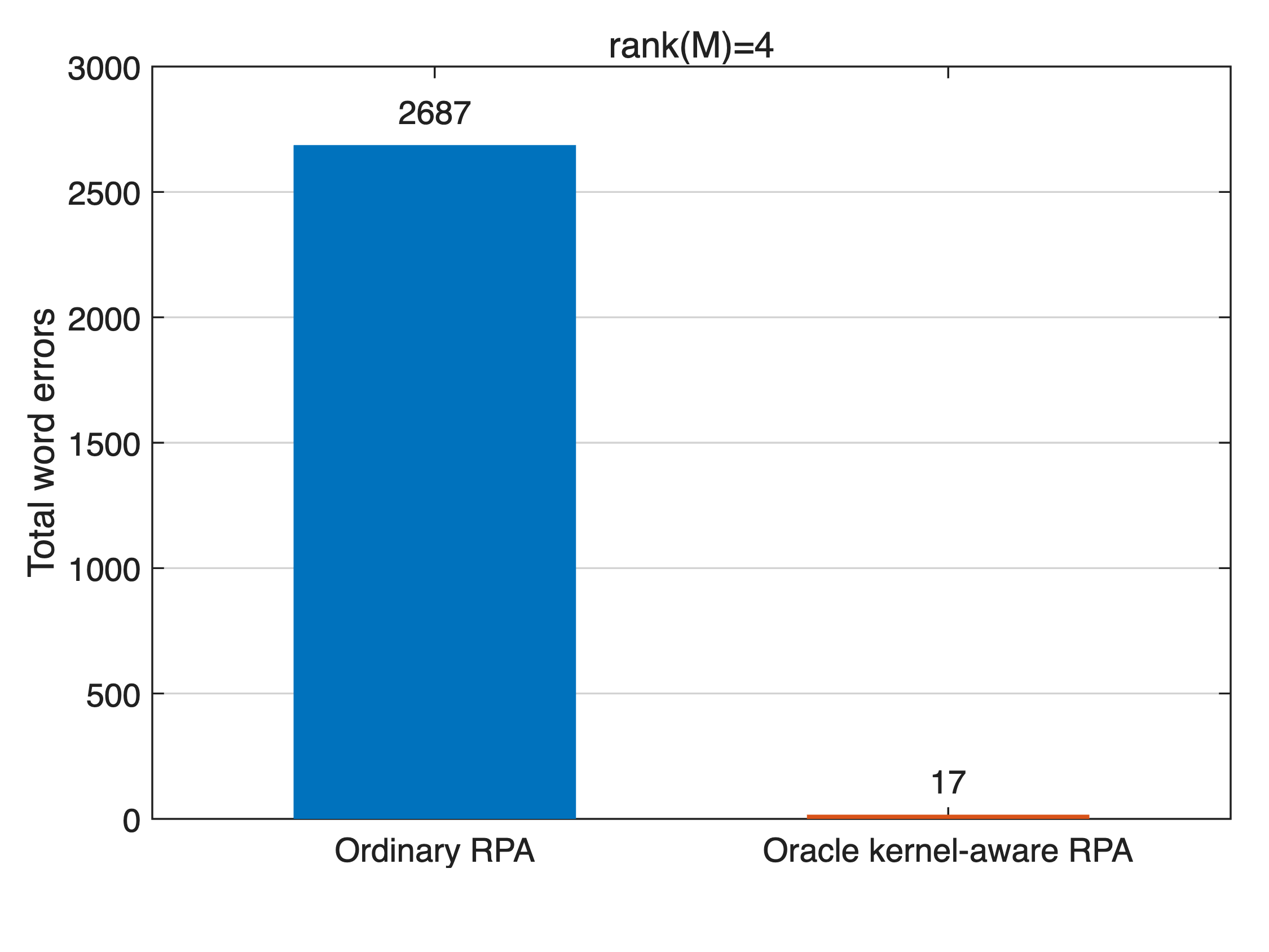}
        \vspace{0.2cm}
        {\scriptsize Fig.~\ref{fig:kernel-aware-rpa}b.
        $\operatorname{rank}(M)=4$}
    \end{minipage}

    \vspace{0.05cm}

    \begin{minipage}[t]{0.33\textwidth}
        \centering
        \includegraphics[width=\linewidth]
        {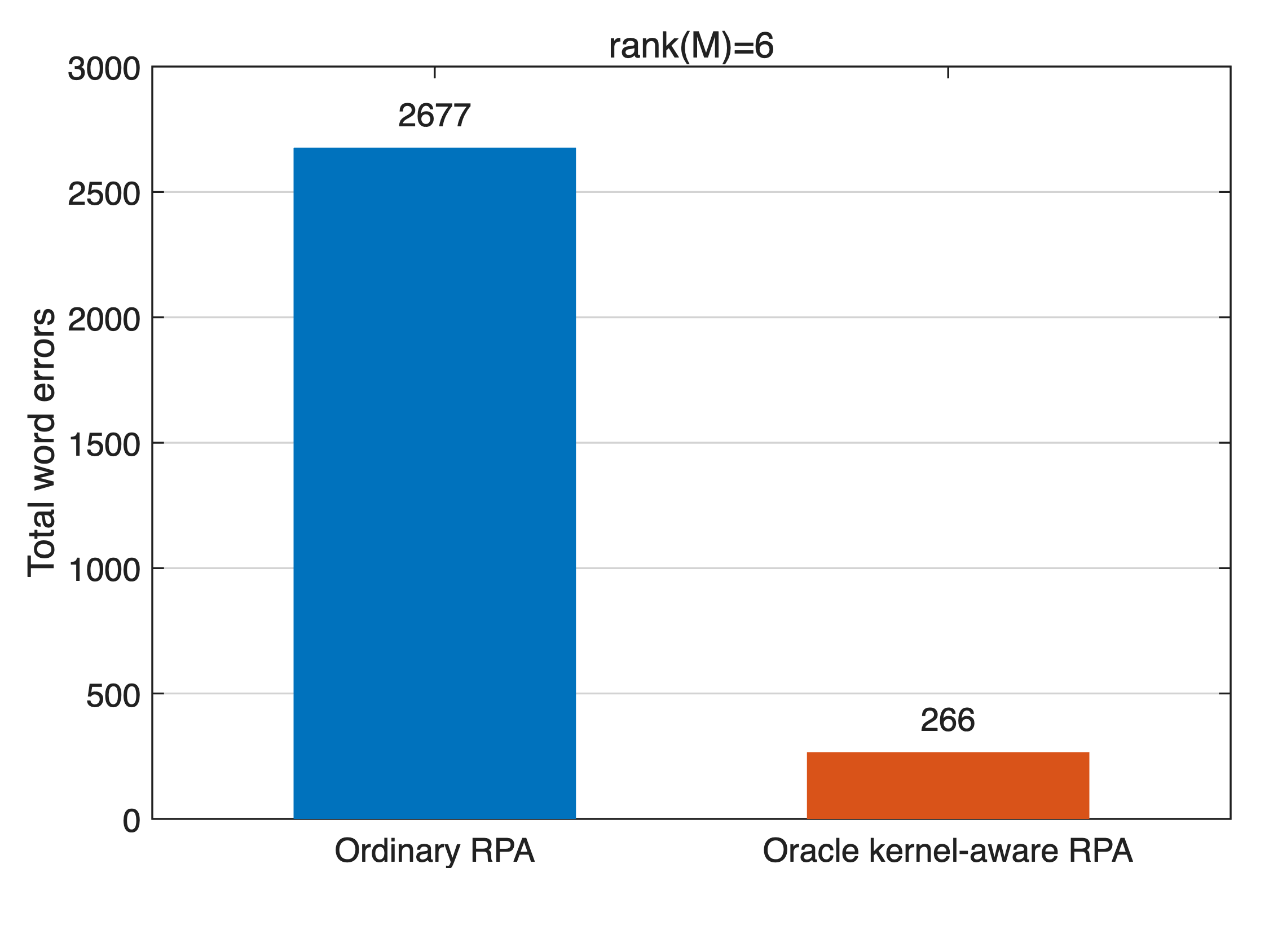}
        \vspace{0.2cm}
        {\scriptsize Fig.~\ref{fig:kernel-aware-rpa}c.
        $\operatorname{rank}(M)=6$}
    \end{minipage}
    \hspace{0.04\textwidth}
    \begin{minipage}[t]{0.33\textwidth}
        \centering
        \includegraphics[width=\linewidth]
        {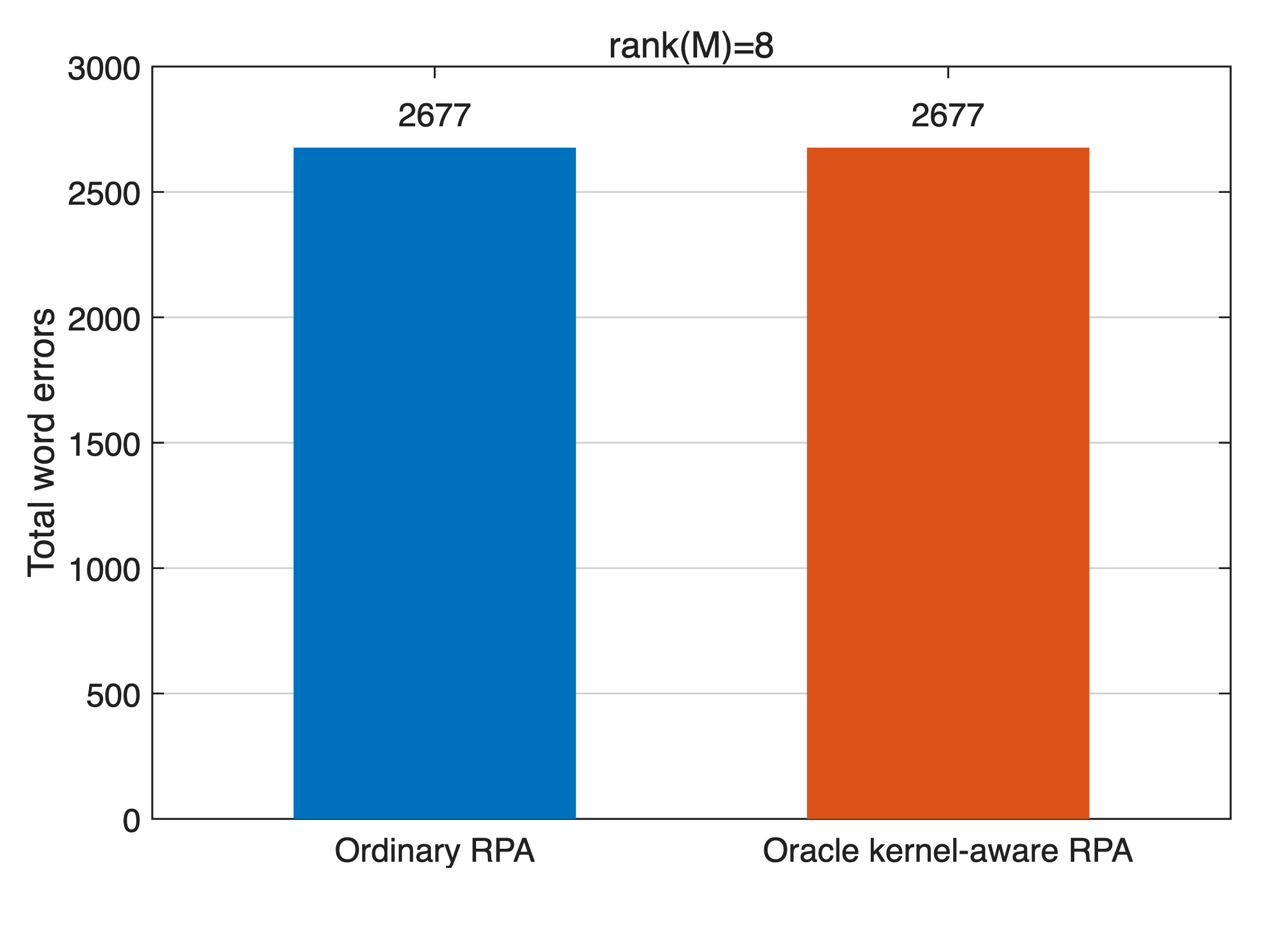}
        \vspace{0.2cm}
        {\scriptsize Fig.~\ref{fig:kernel-aware-rpa}d.
        $\operatorname{rank}(M)=8$}
    \end{minipage}

    \vspace{0.1cm}
    \caption{Comparison of ordinary RPA and kernel-aware RPA.}
    \label{fig:kernel-aware-rpa}
\end{figure}
Figure~\ref{fig:kernel-aware-rpa} shows that the benefit of the modified decoder increases with the number of duplicate measurements. When
$\operatorname{rank}(M)=8$, the null space is trivial and each linear component is measured only once, so the two decoders have identical error counts. As the rank decreases, the number of measurements per linear component grows, and the error count of the modified decoder falls from $266$ to $0$ observed errors; ordinary RPA remains near $2670$ errors.\\
\hfill\break
These results show that duplicate measurements provide valuable information
for aggregation when their equivalence structure is known. They also indicate
that ordinary RPA leaves some of this information unused because it does not
explicitly combine projections that measure the same linear component. Our results also indicate that the projection and recursion stage of CHIRRUP are able to gather this repeated evidence more efficiently. \\

\section{Conclusion}
We have connected the problem of recovering a noisy quadratic form to the problem of recovering a noisy $\Z_4$-linear quadratic form. We have described how Euclidean distance between evaluation vectors is determined by the rank of the difference between the two forms. We have analyzed the performance of the RPA and CHIRRUP algorithms and shown that it depends, not only on geometry, but on the availability of multiple measurements.\\
\hfill\break
We have shown that RPA performance can be improved if the radical of the quadratic form is known. We do not know if determining the radical is easier than determining the quadratic form, and we leave this question to future research.
\newpage
\printbibliography
\newpage
\appendix

\section{Proof of Theorem~\ref{Projection}}
\label{app:rpa-rank}

\begin{proof}
Let $V=\Z_2^{m+1}$, and let $Q_M:V\to \Z_2$ be a binary quadratic form with
skew-symmetric matrix $M$. By definition, for every $x,b\in V$,
\begin{equation}
Q_M(x+b)+Q_M(x)=x^T M b+Q_M(b).
\end{equation}
Therefore, the RPA projection in the direction $b$ is
\begin{equation}
D_bQ_M(x)
=
Q_M(x+b)+Q_M(x)
=
x^T M b+Q_M(b).
\end{equation}
This is a first-order Reed-Muller word. Its linear part is $x^T M b$, and its
constant term is $Q_M(b)$.\\
We also check that this projection is well-defined on the cosets of
$\langle b\rangle$. Replacing $x$ by $x+b$ gives
\begin{equation}
(x+b)^T M b
=
x^T M b+b^T M b.
\end{equation}
Since $M$ is skew-symmetric over $\Z_2$, we have $b^T M b=0$. Hence $x^T M b$
is constant on each coset of $\langle b\rangle$, so $D_bQ_M$ is well-defined on
the quotient space.\\
It remains to count the possible linear parts. The map
\begin{equation}
    b \to Mb
\end{equation}
is linear. Since $\operatorname{rank}(M)=m+1-2r$, its image has dimension $m+1-2r$. By rank-nullity, its kernel has dimension
\[
    (m+1)-(m+1-2r)=2r.
\]
Therefore each element of the image has exactly $2^{2r}$ preimages. In particular, the number of directions satisfying $Mb=0$ is $2^{2r}$, and so
\begin{equation}
    \Pr_b[Mb=0]
    =
    \frac{2^{2r}}{2^{m+1}}
    =
    2^{-(m+1-2r)}.
\end{equation}
This proves Theorem III.1.
\end{proof}
\end{document}